\documentclass[sigconf]{acmart} 

\AtBeginDocument{%
  }

\setcopyright{acmlicensed}
\copyrightyear{2024}
\usepackage{paralist}
\acmYear{2024}
\acmDOI{XXXXXXX.XXXXXXX}

\acmConference[ICCPS '25]{}{}{}
\acmISBN{978-1-4503-XXXX-X/18/06}

\usepackage{bm}
\usepackage{tikz}
\usetikzlibrary{shapes,arrows,fit,backgrounds,automata}
\usetikzlibrary{shapes.geometric,fit}
\tikzstyle{container} = [draw, rectangle, inner sep=0.3cm]

\usepackage{acronym}
\usepackage{graphicx}
\usepackage{subcaption}
\usepackage[commandnameprefix=ifneeded]{changes}
\usepackage{amsmath}

\newcommand{\ie}{\textit{i.e.}}

\newcommand{\dist}{\mathcal{D}}

\newcommand{\supp}{\mathsf{Supp}}

\newcommand{\reals}{\mathbb{R}}

\newtheorem{problem}{Problem}

\newtheorem{remark}{Remark}

\acrodef{pomdp}[POMDP]{partially observable Markov decision process}
\acrodef{mdp}[MDP]{Markov decision process}
\acrodef{asw}[ASW]{Almost-Sure Winning}
\acrodef{cps}[CPSs]{Cyber-Physical Systems}
\acrodef{hmm}[HMM]{hidden Markov model}
\acrodef{des}[DES]{discrete-event system}

 \DeclareMathOperator*{\optmax}{\mathrm{maximize}}

 \DeclareMathOperator*{\optst}{\mathrm{subject\ to}}

  \newcommand{\Expect}{\mathbb{E}}
   \newcommand{\probs}{\mathbb{P}}

\begin{document}

\title{Synthesis of Dynamic Masks for Information-Theoretic Opacity in Stochastic Systems}



\author{Sumukha Udupa}
\affiliation{%
  \institution{Department of Electrical and Computer Engineering, University of Florida}
  \city{Gainesville}
  \country{US}}
\email{sudupa@ufl.edu}

\author{Chongyang Shi}
\affiliation{%
  \institution{Department of Electrical and Computer Engineering, University of Florida}
  \city{Gainesville}
  \country{US}}
\email{c.shi@ufl.edu}

\author{Jie Fu}
\affiliation{%
  \institution{Department of Electrical and Computer Engineering, University of Florida}
  \city{Gainesville}
  \country{US}}
\email{fujie@ufl.edu}








\begin{abstract}
In this work, we investigate the synthesis of dynamic information releasing mechanisms, referred to as ``masks'', to minimize information leakage from a stochastic system to an external observer. Specifically, for a stochastic system, an observer aims to infer whether the final state of the system trajectory belongs to a set of secret states. The dynamic mask seeks to regulate sensor information in order to maximize the observer's uncertainty about the final state, a property known as final-state opacity. While existing supervisory control literature on dynamic masks primarily addresses qualitative opacity, we propose quantifying opacity in stochastic systems by conditional entropy, which is a measure of information leakage in information security. We then formulate a constrained optimization problem to synthesize a dynamic mask that maximizes final-state opacity under a total cost constraint on masking. To solve this constrained optimal dynamic mask synthesis problem, we develop a novel primal-dual policy gradient method. Additionally, we present a technique for computing the gradient of conditional entropy with respect to the masking policy parameters, leveraging observable operators in hidden Markov models.  
To demonstrate the effectiveness of our approach, we apply our method to an illustrative example and a stochastic grid world scenario, showing how our algorithm optimally enforces final-state opacity under cost constraints.

\end{abstract}

\ccsdesc[500]{Security and privacy~Information flow control}
\ccsdesc[500]{Computer systems organization~Sensor networks}

\keywords{Dynamic mask,  opacity, information leakage, hidden Markov model. }

\maketitle

\section{Introduction}
Networked robotics and IoT \ac{cps} have become more prolific and capable in recent years, enabling them to perform tasks in increasingly open environments.  As a result, these systems pose significant security and privacy risks to users. 
Recent studies \cite{JimCottons,ZackWhittaker} have highlighted concerns about inadvertent leakage of location data through widely used mobile applications, such as Strava and Fitbit, which track users' activity and exercise routines. Despite features to obscure details like trip origins and destinations, they can still to expose sensitive information. This is particularly problematic when the users include individuals in high-security roles, such as law enforcement or government personnel. 



Motivated by these security incidents, this paper studies how to synthesize a dynamic sensor information releasing policy, referred to as a \emph{dynamic mask}, to enforce information security and opacity of \ac{cps}  by strategically selecting which sensors can release information to the public or unauthorized observers.  
Opacity has been proposed to generalize secrecy, anonymity, privacy, and other confidentiality properties against attacks \cite{mazare2004using,zeng2019quantitative,watson2012multi,jacob2016overview}. A system is  \emph{opaque} if its secret or private behaviors cannot be reliably deduced by an observer with partial observations.
In this context, we model a stochastic system as a Markov chain, whose state is partially observable to an external intruder (an observer). The system is to optimize a state-based opacity, measured by the uncertainty of the observer about whether the last state of a finite trajectory is in a set of secret states. 
The observer has prior knowledge about the secret based on the knowledge of the system dynamics. To optimize opacity, the system can dynamically change the observer's observation function, by masking sensors accessible to the observer. To illustrate the application scenario, consider a self-driving taxi navigating through an urban environment. If an external observer with access to the traffic camera network can infer the car's final state - whether it has reached a specific destination, then the privacy of the user is compromised. 
However, dynamic masking can be leveraged for privacy protection by determining at runtime which images about the vehicles shall be masked from the camera. 

To this end, we introduce quantitative opacity measured by the conditional entropy of the secret given the observer's partial observations. Since dynamic masking often comes with  data management and operational costs, we are interested in addressing the following question: ``How to synthesize a budget-constrained dynamic masking policy that minimizes information leakage   with respect to a given  secret,  in the presence of a well-informed observer?'' Here, a well-informed observer has full knowledge of the system dynamics, the observation function, and the dynamic mask.


\subsection{Related Works}
Opacity was first introduced by Mazar\'e \cite{mazare2004using}  in the context of cryptographic protocols and has since been expanded to address the security of various types of secrets. For example, state-based opacity ensures that an observer cannot determine whether a secret state has been reached, while language-based opacity prevents the observer from discerning whether the system's execution belongs to a set of secret trajectories \cite{berard2015probabilistic, bryans2008opacity, saboori2010opacity}. In system and control literature,
opacity has been extensively studied in supervisory control of \ac{des}s, with a \emph{qualitative} measure: A system is \emph{qualitatively} opaque if an observer, with partial observations of the system, cannot infer any secret information with certainty \cite{jacob2016overview, saboori2010verification, saboori2011opacity, wu2013comparative, yin2019infinite}. In an opaque system, a trajectory satisfying the secret property must be observation-equivalent to a trajectory that violates it.
For stochastic systems, existing work employs probabilistic opacity, which measures the probability that a secret may be disclosed \cite{saboori2010opacity, saboori2013current, keroglou2018probabilistic} or quantifies the security level of a system by the probability of generating an opaque trajectory \cite{berard2015probabilistic, udupa2024planning}. A comprehensive review of various notions of opacity and their enforcement techniques is provided \cite{jacob2016overview}. 

The enforcement of qualitative opacity in \ac{des}s can be broadly classified into two approaches:    control design to restrict system behavior or dynamic information releasing to limit the observer's observations. In the first approach, opacity is enforced by designing a supervisory controller \cite{saboori2010verification} or by solving an opacity-enforcement game \cite{helouet2018opacity,udupa2023opacity}. 
However, it may not always be feasible or practical to modify or restrict the system's behavior. For example, in the case of mobile applications unintentionally disclosing personal information, it is not practical to restrict a user's movement or prevent them from visiting certain locations. Therefore, an alternative approach explored in the literature involves altering the public observations relayed to the observer. 

This alternative line of work is most closely related to our work and focuses on enforcing opacity by restricting/altering the observer's observations. In  \cite{cassez2012synthesis}, the authors introduce the concept of a ``mask''   to limit the system's observable outputs, either statically or dynamically. A dynamic mask changes   the external observer's observation mapping at each execution step for opacity enforcement. Dynamic masks have been designed to enforce current-state opacity \cite{cassez2012synthesis}   and infinite-state opacity\cite{yin2019infinite}. Alternatively,  the work\cite{ji2018enforcement,ji2019opacity} developed  selective insertion and/or deletion of output observations for opacity enforcement. 

A dynamic mask can be viewed as an information-flow control mechanism that  regulates the information to the observer by   enabling or disabling the associated sensors \cite{yin2019synthesis}. Closely related to dynamic masking,  dynamic sensor activation has been extensively studied in the context of \ac{des}s \cite{sears2016minimal,yin2019general}, for  fault diagnosis \cite{cassez2008fault,thorsley2007active,wang2010minimization} and detectability \cite{shu2013online}. 
The work \cite{zhang2015maximum} studied the problem of maximum information release while ensuring the opacity of \ac{des}. More recently, the authors \cite{wintenberg2022dynamic} presented a dynamic obfuscation framework that enables an intended recipient to infer sensitive information while preventing unintended recipients from making similar deductions.  

Building on the insight from qualitative opacity enforcement using dynamic masks and dynamic sensor activation, our work distinguishes itself by focusing on \emph{quantitative, information-theoretic} definitions of opacity \cite{shi2024information} and optimal opacity enforcement under cost constraints. Specifically, we use conditional entropy as a measure of information leakage, subject to a generic cost associated with masking. The use of conditional entropy as the measure of information leakage allows us to have a symmetric notion of opacity, \ie, the opacity will be minimal when the observer is consistently confident that the agent either visited or avoided the secret states given an observation. 
This sets our approach apart from qualitative opacity-enforcement approaches \cite{cassez2012synthesis,yin2019general}. 

In addition to  opacity,  differential privacy has been widely studied in \ac{cps} \cite{hassan2019differential}. It protects privacy by adding calibrated noise, balancing privacy and data accuracy. 
Our approach differs fundamentally from differential privacy in that it uses different measures of information leakage and leverages existing   noises in the stochastic environment dynamics and imperfect observation channels for opacity, rather than introducing noise to the system.


\subsection{Contributions}
Our contributions to this work are as follows:
\begin{enumerate}
    \item We introduce conditional entropy as a novel measure of information leakage regarding the secret property of a stochastic system, given an observer's partial observations. This metric quantifies the uncertainty about the secret state, providing a framework for analyzing the opacity of a stochastic system or an observer with imperfect observations. 
    \item We propose a new method to compute the approximate gradient of conditional entropy with respect to the parameters of a dynamic masking policy. This method leverages the observable operators in \ac{hmm}s. Using the computed gradient, we formulate the problem of enforcing maximal opacity with dynamic masking under cost constraints as a constraint optimization problem. We then employ primal-dual gradient-based optimization to compute a locally optimal policy that maximizes opacity while satisfying the cost constraint on the dynamic masks.
     \item We demonstrate the effectiveness of our proposed algorithm through experiments and share our insight into how the dynamic mask protects the opacity in a stochastic system.
\end{enumerate}

\section{Preliminaries and Problem Formulation}
\label{sec:preliminaries_and_problem_formulation}

\paragraph{Notations} Given a finite set $Z$, let $\dist(Z)$ be the set of all probability distributions over $Z$. Given a distribution $d\in \dist (Z)$, let $\supp(d) = \{z\in Z\mid d(z)>0\} $ be the support of this distribution. The set $Z^T$ denotes the set of sequences with length $T$ composed of elements from $Z$, $Z^{\le  T}$ denotes the set of sequences with length $\le T$ and $Z^\ast$ denotes the set of all finite sequences generated from $Z$. $\mathbb{R}$ and $\mathbb{N}$ represent the real and natural numbers respectively. $\mathbb{P}$ represents the probability measure. We denote random variables by capital letters and their realization by lowercase \ie, $X$ and $x$. The sequence of random variables and their realizations of length $T$ are denoted as $X_{[0:T]}$ and $x_{[0:T]}$ respectively.

We consider the interaction between two agents in a stochastic system: an observer and a masking agent. The observer aims to infer,  from his partial observation of the system,  if the last state of a finite trajectory is a secret.
A collection of sensors generates partial observations. The masking agent aims to determine dynamically what sensor readings can be released or masked to the observer, such that the observer is maximally uncertain about the secret.

The system under partial observations is modeled by an \ac{hmm} with a controllable emission function. 
\[
M = \langle S,  P,  O,  \Sigma, \mu_0, \sigma_0, E \rangle
\]
where
\begin{itemize}
    \item $S$ is a finite set of states;
     \item $P:S \rightarrow \dist(S)$ is the probabilistic transition function such that for each state $s,s'\in S$, $P(s'\mid s)$ is the probability of reaching $s'$   from  state $s$; 
      \item $O $ is the finite set of observations;
      \item $\Sigma$ is a finite set of sensor configurations;
      \item $\mu_0$ is the initial state distribution, such that $\mu_0(s)$ represents the probability of $s$ being the initial state;
      \item $\sigma_0$ is the initial sensor configuration;
   \item $E: S\times   \Sigma \rightarrow \dist(O)$ is the emission function that maps each state and a sensor configuration to a distribution over observations. The emission function is \emph{controllable} as the system can determine the sensor configuration.
\end{itemize}

 A randomized \emph{state-based, Markovian dynamic mask} is a function $\pi: S \times\Sigma \rightarrow \dist(\Sigma)$
  that maps a current state $s$, a current sensor configuration $\sigma$, to a distribution over next sensor configurations. 
  We also refer to the set $\Sigma$ of sensor configurations as \emph{the set of masking actions}.
  
A dynamic mask influences the observation of a state by changing the sensor configuration. It may determine which sensors can release information and which sensors cannot.
  A dynamic mask is associated with a cost of masking: $C: S \times \Sigma \times \Sigma \rightarrow \reals$ that maps a state, the current sensor configuration, and the next sensor configuration to a non-negative cost $C(s,\sigma,\sigma')\ge 0$.

  Consider a fixed dynamic mask $\pi$, the system starts with the initial state $s_0\sim \mu_0$ and initial sensor configuration $\sigma_0$; the observer receives initial observation $o_0$ with probability $E(o_0|s_0,\sigma_0)$.
  At the next time step,  the state $s_1$ is reached with probability $P(s_1|s_0)$, and the sensor configuration $\sigma_1$ is selected with probability $\pi(\sigma_1| s_0,\sigma_0)$. 
  Given the state $s_1$ and sensor configuration $\sigma_1$, the observation $o_1$ is emitted with probability $E(o_1|s_1,\sigma_1)$. This iteration continues for a finite horizon $T$. By the end of the horizon, the state trajectory is $s_0s_1\ldots s_T$, the observation sequence is $o_0o_1\ldots o_T$ and the sequence of sensor configurations under the dynamic mask is $\sigma_0\sigma_1\ldots \sigma_T$.

  \begin{remark}
      Note that our definition of the observations and emission function is generic to handle cases of masking actions visible and invisible to the observer. In the case of masking actions visible to the observer, the observations include current masking action. 
  \end{remark}

We introduce entropy as a measure of the observer's uncertainty regarding the secret.
 The entropy of a random variable $X$ with a countable support $\mathcal{X}$ and a probability mass function $p$ is 
\[
H(X) = -\sum_{x\in \mathcal{X}}p(x)\log p(x).
\]
The higher the entropy the larger the level of uncertainty in the random variable $X$.
Note that the $\log$ is the base-two logarithm.

The conditional entropy of $X_2$ given $X_1$ is 
\[
H(X_2\mid X_1) = -\sum_{x_1\in \mathcal{X}}\sum_{x_2 \in \mathcal{X}}p(x_1,x_2)\log p(x_2\mid x_1).
  \]

Next, we introduce a definition of state-based opacity based on quantitative information leakage \cite{smithFoundationsQuantitativeInformation2009,khouzaniLeakageMinimalDesignUniversality2017} in information security.

Given an \ac{hmm} with controllable emission function $M$, 
a dynamic mask $\pi$ induces a stochastic process $M^\pi\coloneqq \{S_t, \Sigma_t, O_t, t\ge 0\}$ where $S_t$, $\Sigma_t$, and $O_t$  are the state, sensor configuration, and observation at time $t$. Let $\probs^\pi$ be the probability measure induced by dynamic mask $\pi$.
\begin{definition}[State-based final state opacity]
   Given the stochastic process $M^\pi \coloneqq \{S_t, \Sigma_t, O_t, t\ge 0\}$, a finite horizon $T>0$, and a set $G$ of secret states, for any $0\leq t\leq T$, let the random variable $W_t$ be defined by 
    \[
    W_t = \mathbf{1}_G(S_t).
    \]
    That is, $W_t$ is the Boolean random variable representing if the state $S_t$ at time $t$   is in the secret set $G$.
The \emph{quantitative final-state opacity} under dynamic mask $\pi$ given horizon $T$ is defined by 
\begin{multline}
H(W_T|O_{0:T};  \pi) =\\
-\sum_{W_T \in \{0,1\}} \sum_{o_{0:T}\in O^T }  \probs^\pi (W_T,o_{0:T})  \cdot \log \probs^\pi (W_T\mid o_{0:T} ),
\end{multline}
 which is the conditional entropy of  $W_T$ given     observation $O_{0:T}$. 
\end{definition}

\begin{remark}
Traditional supervisory control employs the notion of qualitative current (or final) state opacity \cite{saboori2007notions, saboori2011opacity}, which stipulates that the observation of a trajectory satisfying the secret property (\ie, the final state is in the set $G$) is observation-equivalent to a trajectory that does not satisfy it (\ie, the final state is not in the set $G$).  Equivalent non-secret behavior ensures that the observer cannot distinguish whether the current (or final) state belongs to the set of secret states. In contrast, quantitative opacity offers a symmetrical measure that quantifies the observer's confidence in identifying whether the agent is in a secret state. This is achieved by comparing the observer's ability to infer the secret before and after analyzing the available observations.
\end{remark}

We can now formulate our problem as a problem of designing a dynamic masking policy to maximize final state opacity.


\begin{problem}[Cost-constrained optimal dynamic mask for final-state opacity]
\label{problem:1}
    Given the \ac{hmm} $M$, a set of secret states $G$,  a finite horizon $T$, compute a dynamic mask $\pi$ that maximizes the conditional entropy $H(W_T\mid O_{0:T}; \pi)$ of the random variable $W_T$ given the observation sequence $O_{0:T}$, while ensuring the total expected discounted cost of masking does not exceed a given threshold $\epsilon$.
    \begin{align}
        \label{eq:optimization-mask}
& \optmax_\pi \quad
H(W_T|O_{0:T};  \pi)  \nonumber \\
\optst & \quad \sum_{t=0}^{T-1} \Expect_\pi \left[ \gamma^t C(S_t,\Sigma_t,\Sigma_{t+1})\right] \le \epsilon.
    \end{align}
    where $\gamma \in [0,1]$ is the discounting factor and the expectation is taken with respect to the mask $\pi$-induced stochastic process $M^\pi$.
\end{problem}
In other words, despite knowing the hidden Markov model and the dynamic mask, the observer is maximally uncertain regarding whether the final state is a secret state.  While achieving the maximum opacity through dynamic masking, it is also required for the masking agent to ensure that the total expected discounted cost of masking does not exceed the given budget (\ie, no more than $\epsilon$).

We use the following running example to illustrate our definitions and problem formulation.

\begin{example}[Part I]
    \label{ex: 1}
    Consider the \ac{hmm} in Figure~\ref{fig:running_example} with 7 states, $s_0$ through $s_6$ with the initial state $s_0$ and a sink state $s_5$. Here,  the set of secret states is $G=\{s_4,s_6\}$. The transitions are labeled with the transition probabilities. 
   Observations are obtained through the sensors deployed in the environment. Sensors $R$, $G$, $P$, and $B$ monitor the states $s_1$, $s_3$, $s_4$, and $s_6$ respectively. The sensors relay the sensor labels ($R$, $G$, $P$, or $B$) if the current state is in the coverage of the respective sensor and a null observation `0' otherwise. The sensors have a false negative rate of $0.15$ and a false positive rate of $0$. The masking agent can mask at most one sensor at a time. The masking action is visible to the observer. 
   Thus, the set of sensor configurations can be defined by $\Sigma = \{R,G, P, B, N\}$ where configuration $X \in \{R,G,P,B\}$ means that sensor $X$ is masked, and configuration $N$ means none of the sensors are masked. 
  The product of sensor labels and masking actions form the set of observations. 

  In this example, we consider the initial masking configuration to be $N$, that is, no masking. We consider the following cost structure for masking the sensors. 
  Let $\phi:\Sigma\to \reals$ be a mapping from sensor configurations to a real number cost. In this example we have, $\phi(sensor) = 10$ for $sensor \in \{R, G, P\}$, and $\phi(B) = 30$, then,
    \[
     C(s,\sigma, \sigma') = \begin{cases}
         \phi(\sigma') &\text{if } \quad \sigma \neq \sigma', \sigma'\in \Sigma \setminus\{N\},\\
         \frac{\phi(\sigma')}{2} & \text{if } \quad \sigma = \sigma', \sigma' \in \Sigma \setminus \{N\},\\
       0& \text{ otherwise.}
        \end{cases}
  \]

  Intuitively, masking any sensor other than sensor $B$ costs $10$, masking sensor $B$ costs $30$, and continuing to mask the same sensor that was masked at the previous time step costs half the set cost. Finally, not masking any sensor has no cost associated with it.
    
The cost threshold is given by $\epsilon$. In this illustrative example, we consider a horizon $T=2$. For instance, consider the dynamic mask $\pi(s_0) = R, \pi(s_1) = P, \pi(s_2) = N, \pi(s_3) = B$. A  trajectory in this HMM, for example, $s_0 s_3 s_6$, can generate an observation sequence to P2, $(\text{`0'}, N) (G, R) (\text{`0'}, B)$ with some positive probability. At state $s_0$ and action $R$,  The observation $(\text{`0'}, N)$ means that a null observation is obtained, and none of the sensors are masked.

\end{example}

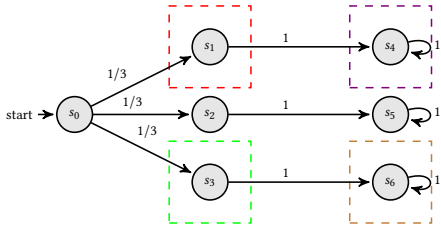
\begin{figure}[h]
    \centering   
   \begin{tikzpicture}[->,>=stealth',shorten >=1pt,auto,node distance=2.5cm,scale=0.6,semithick, transform shape]
    \tikzstyle{every state}=[fill=black!10!white];

    \node[state, initial] (0) at (0, 0.5) {$s_0$};
    \node[state] (1) at (3, 2) {$s_1$};
    \node[state] (2) at (3, 0.5) {$s_2$};
    \node[state] (3) at (3, -1) {$s_3$};
    \node[state] (4) at (7, 2) {$s_4$};
    \node[state] (5) at (7, 0.5) {$s_5$};
    \node[state] (6) at (7, -1) {$s_6$};

    \node [container,fit=(1),draw=red,dashed,line width=0.2mm ] (container) {};
    \node [container,fit=(3),draw=green,dashed,line width=0.2mm ] (container) {};
    \node [container,fit=(4),draw=violet,dashed,line width=0.2mm ] (container) {};
    \node [container,fit=(6),draw=brown,dashed,line width=0.2mm ] (container) {};

      \path 
        (0) edge []  node[pos=0.4]{$1/3$} (1)
        (0) edge []  node[pos=0.4]{$1/3$} (2)
        (0) edge []  node[pos=0.4]{$1/3$} (3)

        (1) edge []  node[pos=0.4]{$1$} (4)

        (2) edge []  node[pos=0.4]{$1$} (5)

        (3) edge []  node[pos=0.4]{$1$} (6)

        (4) edge [loop right]  node[pos=0.4]{$1$} (4)
        (5) edge [loop right]  node[pos=0.4]{$1$} (5)
        (6) edge [loop right]  node[pos=0.4]{$1$} (6)

        ;

\end{tikzpicture} 
    \caption{\ac{hmm} for the illustrative example. The colored boxes represent the sensors in the environment. Red - sensor $R$, Green - sensor $G$, Purple - sensor $P$, and Brown - sensor $B$.}
    \label{fig:running_example}
\end{figure}

\section{Main Results}


    

\subsection{Synthesizing A Dynamic Mask}

We first show that the planning problem with a dynamic mask can be formulated as the following \ac{mdp}.
\begin{definition}
   The \ac{mdp} for the dynamic mask is a tuple
    \[
     \mathcal{M} = \langle \mathcal{Z}, \Sigma,\mathcal{O}, \mathbf{P},\mathbf{E},\mathcal{C} ,\mu_0 \rangle
    \]
    where,
    \begin{itemize}
    
        \item $\mathcal{Z} =S\times \Sigma$ is the set of states.
        \item $\Sigma$ is the masking action space, \ie, the set of all possible sensor configurations. 
        \item $\mathcal{O}$ is the finite set of observations.
        \item $\mathbf{P}:\mathcal{Z}\times \Sigma \to \dist(\mathcal{Z})$ is the probabilistic transition function, defined as follows:  for every pair of states   $z=(s,\sigma)$ and $z'=(s', \sigma')$, and a masking action $\sigma^o \in \Sigma$,
   \[     \mathbf{P}(z'\mid z, \sigma^o) =  \begin{cases}
         P(s'\mid s) &\text{if }\sigma' = \sigma^o\\
       0& \text{ otherwise.}
        \end{cases}
       \]

        \item $\mathbf{E}:\mathcal{Z}\to \dist(\mathcal{O})$ is the emission function (observation function) that maps a state $z=(s, \sigma)$ to a distribution over $o\in \mathcal{O}$ given by
        \[
        \mathbf{E}(z) = E(s, \sigma).
        \]

        \item $\mathcal{C}:\mathcal{Z}\times \Sigma\to \mathbb{R}$ is the transition-based cost function that defines the cost of taking a masking action at a state and is given by
        \[
        \mathcal{C}(z, \sigma') = C(s, \sigma, \sigma'),
        \]
        for every $z=(s, \sigma)\in \mathcal{Z}$ and $\sigma'\in \Sigma$.

        \item $Z_0$ is the initial state, which is a random variable with the following distribution.
        \[
        \mathbb{P}(Z_0 = (s,\sigma_0)) = \mu_0(s),
        \]
        for each $s \in S$ and $\sigma_0$ the initial sensor configuration. And with a slight abuse of notation, we use $\mu_0(z)$ to represent the probability that $z$ is the initial state.

    \end{itemize}
\end{definition}

 P1's value function in the above \ac{mdp}, given a dynamic mask $\pi$, is $V^\pi:\mathcal{Z}\to \reals$, which is defined as follows: for any $z$,  
\[
V^\pi(z) = \mathbb{E}_\pi[\sum_{k=0}^\infty \gamma^k \mathcal{C}(Z_k, \pi(Z_k))\mid Z_0=z],
\]
where $\mathbb{E}_\pi$ is the expectation with respect to the probability distribution induced by the dynamic mask $\pi$ from $\mathcal{M}$. And $Z_k$ is the $k-$th state in the Markov chain induced by the dynamic mask from $\mathcal{M}$.

\begin{example} [Part II]
    Continuing Example \ref{ex: 1}, a fraction of the \ac{mdp} $\mathcal{M}$ is shown in Figure~\ref{fig:running_example_mdp}. The states in $\mathcal{M}$ are the states of \ac{hmm} augmented with the current masking configuration of sensors, \ie, in state $(s_0, N)$, the first element represents the state and the next element the masking action. The transitions are labeled with the next masking action and the transition probability. The observation distribution at each state is obtained given the augmented state. For instance, at the state $(s_1, R)$,  P2 has a null observation with probability $1$ (\ie, (`0', $R$) with probability 1), and all other observations with probability $0$. Likewise, at the state $(s_3, R)$, P2 has a null observation (\ie (`0', $R$)) with probability $0.15$ and the observation $G$ with probability $0.85$ (\ie, ($G$, $R$)). 

\end{example}

\begin{figure}[h]
    \centering   
   \begin{tikzpicture}[->,>=stealth',shorten >=1pt,auto,node distance=2.5cm,scale=0.6,semithick, transform shape]
    \tikzstyle{every state}=[fill=black!10!white];

    \node[state, initial] (0) at (0, 0.5) {$(s_0, N)$};
    \node[state] (1) at (3, 2) {$(s_1, R)$};
    \node[state] (2) at (3, 0.5) {$(s_2, R)$};
    \node[state] (3) at (3, -1) {$(s_3, R)$};
    \node[state] (4) at (7, 2) {$(s_4, P)$};
    \node[state] (5) at (7, 0.5) {$(s_5, B)$};
    \node[state] (6) at (7, -1) {$(s_6, P)$};


      \path 
        (0) edge []  node[pos=0.4, sloped]{$\sigma = R$, $(1/3)$} (1)
        (0) edge []  node[pos=0.4, sloped]{$\sigma = R$, $(1/3)$} (2)
        (0) edge []  node[pos=0.4, sloped]{$\sigma = R$, $(1/3)$} (3)

        (1) edge []  node[pos=0.4, sloped]{$\sigma=P$, $(1)$} (4)

        (2) edge []  node[pos=0.4, sloped ]{$\sigma=B$, $(1)$} (5)

        (3) edge []  node[pos=0.4, sloped]{$\sigma=P$, $(1)$} (6)

        (4) edge [loop right]  node[pos=0.4]{$\sigma=P$, $(1)$} (4)
        (5) edge [loop right]  node[pos=0.4]{$\sigma=B$, $(1)$} (5)
        (6) edge [loop right]  node[pos=0.4]{$\sigma=P$, $(1)$} (6)

        ;

\end{tikzpicture}
    \caption{A fragment of the \ac{mdp} for the \ac{hmm} in Ex.~\ref{ex: 1}.}    
    \label{fig:running_example_mdp}
\end{figure}
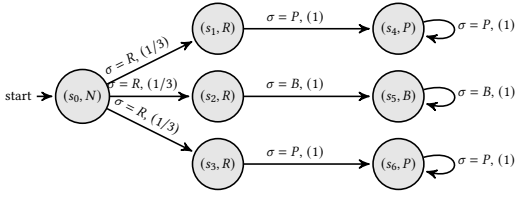

\subsection{Primal-Dual Policy Gradient}
For notational convenience, we introduce an index for the finite-state set $\mathcal{Z}=\{1,\cdots, N\}$, where $N = |S\times \Sigma|$. Consider a set of parameterized dynamic masks $\{\pi_{\theta}\mid \theta \in \Theta\}$ where $\Theta$ is a finite-dimensional parameter space. For any dynamic mask $\pi_\theta$ parameterized by $\theta$, the Markov chain induced by $\pi_\theta$ from $\mathcal{M}$ is denoted $\mathcal{M}^{\pi_\theta}:\{Z_t, O_t, t\geq 0\}$. The Markov chain follows the transition probability $\probs_\theta$ such that the probability of transitioning from state $z$ to $z'$ in one step  is defined by \begin{equation}
    \label{eq:trans}\probs_\theta(z,z')=\sum_{\sigma'\in \Sigma} \mathbf{P}(z'\mid z, \sigma') \pi_\theta(\sigma'\mid z).
\end{equation}

Thus, Problem \ref{problem:1} can now be formally formulated as a constrained optimization problem under the above parametrization as follows:

\begin{equation}
\begin{aligned}
&\optmax_\theta && H(W_T\mid Y; \theta) \\ 
&\optst: && V(\mu_0, \theta) \leq \epsilon.
\end{aligned}
\label{eq:problem_1}
\end{equation}
where $Y=O_{0:T}$ is the finite  sequence of observations. The value function $V(\mu_0,\theta)$ is computed by evaluating the dynamic mask $\pi_\theta$ given the initial state distribution $\mu_0$, \ie, $V(\mu_0, \theta):= V^{\pi_\theta}(\mu_0)$, and $H(W_T\mid Y;\theta):= H(W_T\mid Y;\mathcal{M}^{\pi_\theta})$.

We then formulate this problem with inequality constraints into an unconstrained optimization problem by introducing the following Lagrangian $L(\theta, \lambda)$, and thereby incorporating the constraint into the objective function:
\begin{equation}
\label{eq:lagrangian}
    L(\theta, \lambda) = H(W_T\mid Y; \theta) + \lambda (\epsilon - V(\mu_0,\theta)),
\end{equation}
where $\lambda \geq 0$ is the Lagrange multiplier.

With this, we can now express the original constrained optimization problem as the following max-min problem:
\begin{equation}
    \label{eq:max_min}
    \max_{\theta} \min_{\lambda \geq 0} L(\theta, \lambda).
\end{equation}

We use the primal-dual gradient descent-ascent algorithm such that in each iteration $k$, we have, 

\begin{equation}
\label{eq:primal_dual_gradient_descent_ascent}
    \begin{aligned}
            \theta_{k+1} = \theta_k + \eta \nabla_{\theta} L(\theta, \lambda),\\
        \lambda_{k+1} = \lambda_k - \kappa (\epsilon - V(\mu_0,\theta)),
    \end{aligned}
\end{equation}
where $\eta>0$ and $\kappa>0$ are step sizes. The gradient of the Lagrangian function with respect to $\theta$ is 
\begin{equation}
    \label{eq:gradient_lagrangian}
    \nabla_{\theta} L(\theta, \lambda) = \nabla_{\theta}H(W_T\mid Y;\theta) - \lambda \nabla_{\theta} V(\mu_0,\theta).
\end{equation}

The gradient of the value function can be computed using the standard sampling-based policy gradient algorithms like REINFORCE \cite{sutton1999policy}. 
But, to obtain the gradient of the Lagrangian function, we also need to compute the gradient of the conditional entropy with respect to the dynamic mask parameter $\theta$. As entropy is non-causal and non-cumulative, it cannot be written as a cumulative sum of terms. We thus present a novel method to compute the gradient $\nabla_{\theta}(H(W_T\mid Y; \theta))$.

\subsection{Gradient Computation for the Conditional Entropy}
 Consider a  dynamic mask  $\pi_\theta$,   from the observer's perspective the stochastic process is an \ac{hmm} $\mathcal{M}_0^\theta = (\mathcal{Z},\mathcal{O}, \probs_\theta, \mathbf{E}, \mu_0)$, where $\mathcal{Z}$ is the state space, $\mathcal{O}$ is the observation space, $\probs_\theta$ is the transition function induced by $\pi_\theta$ (\eqref{eq:trans}) and $\mathbf{E} = \{e_i, i\in \mathcal{Z}\}$ is the emission probability distribution where $e_i(o) = E(s_i, \sigma_i)$ for $z_i=(s_i, \sigma_i)$. 

We can now re-write the conditional entropy of $W_T$ given an observation sequence $Y$ as follows:
\begin{equation}
    \label{eq:parametrized_conditional_entropy}
    H(W_T\mid Y; \theta) = -\sum_{y\in \mathcal{O}^T} \sum_{w_T\in \{0,1\}} \probs_\theta(w_T,y) \log \probs_\theta (w_T\mid y).
\end{equation} 
 
Now, we compute the gradient of $H(W_T\mid Y; \theta)$ with respect to the dynamic mask parameter $\theta$. Using the log-derivative trick $\nabla_\theta \probs_\theta(y) = \probs_\theta(y) \nabla_\theta \log \probs_\theta(y)$ and the property of conditional probability, we have


\begin{equation}
\begin{aligned}
\label{eq:conditional_entropy_gradient_calculation}
 &\nabla_\theta H(W_T\mid Y; \theta) \\
= & - \sum_{y \in \mathcal{O}^T} \sum_{w_T \in \{0,1\}} \Big[\nabla_\theta \probs_\theta(w_T, y) \log \probs_\theta(w_T \mid y) \\
&+  \probs_\theta(w_T, y)\nabla_\theta \log \probs_\theta(w_T\mid y)\Big] \\
= & - \sum_{y \in \mathcal{O}^T} \sum_{w_T \in \{0,1\}} \Big[\probs_\theta(y) \nabla_\theta \probs_\theta(w_T\mid y) \log \probs_\theta(w_T \mid y) \\ 
& + \probs_\theta(w_T\mid y) \frac{\probs_\theta(y)}{\probs_\theta(y)} \nabla_\theta \probs_\theta(y) \log \probs_\theta(w_T \mid y) 
\\ &+ \probs_\theta(y)\frac{\nabla_\theta \probs_\theta(w_T \mid y)}{\log 2}\Big] \\
 = & - \sum_{y \in \mathcal{O}^T} \probs_\theta (y) \sum_{w_T \in \{0,1\}} \Big[ \log \probs_\theta(w_T \mid y) \nabla_\theta \probs_\theta(w_T\mid y) \\
& + \probs_\theta(w_T\mid y) \log \probs_\theta(w_T \mid y) \frac{\nabla_\theta \probs_\theta(y)}{\probs_\theta(y)} 
+ \frac{\nabla_\theta \probs_\theta(w_T \mid y)}{\log 2} \Big].
\end{aligned}
\end{equation}

We now present a novel method to compute $\probs_\theta(y)$, $\probs_\theta(w_T\mid y)$, and the gradients $\nabla_\theta \probs_\theta(y)$, $\nabla_\theta \probs_\theta(w_T\mid y)$  for $w_T \in \{0,1\}$ using observable operators. 
 


We first introduce the matrix notation for \ac{hmm}s. We can write the  \emph{reversed} state transition probability matrix as  $T_\theta\in \mathbb{R}^{N\times N},$  with 
$$T_\theta[i,j] =  \probs_\theta(\mathcal{Z}_{t+1}=i\mid \mathcal{Z}_t=j), $$ 
and the observation probability matrix $B \in \mathbb{R}^{M\times N}$, where $M$ is the number of all possible observations, with 
\[B[i,j] = \mathbf{E}(O_t=i\mid \mathcal{Z}_t=j),
\]
where $O_t$ and $\mathcal{Z}_t$ represent the observation and state at time $t$. 

Thus, we have the probability distribution of any sequence of states and observations fully characterized by $T_\theta$, $B$, and $\mu_0$.


\begin{definition}[Observable operator(\cite{jaeger2000observable})]
    Given the \ac{hmm} $\mathcal{M}_0^\theta$,   for any observation at time $t$ given as $o_t$, the observable operator $A_{o_t}^\theta$ is defined as a matrix of size $N\times N$ with its $ij$-th component given as
    \[
    A_{o_t}^\theta[i,j] = \probs_\theta(\mathcal{Z}_{t+1}=i\mid \mathcal{Z}_t=j)\cdot \mathbf{E}(O_t=o_t\mid \mathcal{Z}_t=j),
    \]
    which is the probability of transitioning from state $j$ to $i$ and the probability of emitting an observation $o_t$ at the state $j$. This can further be expressed in the matrix representation as follows
    \begin{equation}
    \label{eq:observable_operators_matrix}
    A_{o_t}^\theta = T_\theta \text{diag}(B_{o_t,1},\cdots,B_{o_t,N}).
    \end{equation}

\end{definition}



Then, the probability of a sequence of observations $o_{[1:T]}$ given a dynamic mask $\pi_\theta$ is shown to be written as a matrix operation of observable operators as follows:
\begin{equation}
    \label{eq:probability_of_observation_sequence}
    \probs_\theta(Y=o_{1:T}) =\mathbf{1}_N^T A_{o_T}^\theta\cdots A_{o_1}^\theta\mu_0 = \mathbf{1}_N^TA_{o_{T:1}}^\theta\mu_0,
\end{equation}
where $\mathbf{1}_N$ denotes an $N\times 1$ column vector of ones, and $\mathbf{1}_N^T$ its transpose. Here, $A_{o_t}^\theta$ incorporates information about one-step observation likelihoods and hidden state transitions for any $1\leq t\leq T$.

For the gradient computation of the conditional entropy, we require $\nabla_\theta \probs_\theta(y)$. This gradient can be calculated using the product rule and the fact that only the observable operators depend on the mask parameter.

\begin{equation}
\begin{aligned}
\label{eq:gradient_of_P_y}
&\nabla_\theta \probs_\theta(y) = \nabla_\theta (\mathbf{1}_N^T A_{o_{T:1}}^\theta \mu_0)\\
&= \mathbf{1}_N^T \nabla_\theta (A_T^\theta A_{T-1}^\theta \cdots A_1^\theta) \mu_0\\
&= \mathbf{1}_N^T \Big(\sum_{k=1}^T (A_T^\theta A_{T-1}^\theta \cdots \frac{\partial A_k^\theta}{\partial \theta} \cdots A_1^\theta) \Big)\mu_0\\
&= \mathbf{1}_N^T \Big(\sum_{k=1}^T (A_T^\theta \cdots \frac{\partial}{\partial \theta}(T_\theta \text{diag}(B_{o_k,1},\cdots, B_{o_k,N})) \cdots A_1^\theta) \Big)\mu_0\\
&= \mathbf{1}_N^T \Big(\sum_{k=1}^T (A_T^\theta \cdots \frac{\partial T_\theta}{\partial\theta_i}\text{diag}(B_{o_k,1},\cdots, B_{o_k,N}) \cdots A_1^\theta) \Big)\mu_0,
\end{aligned}
\end{equation}
where $\frac{\partial T_\theta}{\partial\theta_i}$ represents the gradient of the tensor $T_\theta$ with respect to each parameter in $\theta$.

Next, to compute $\probs_\theta(w_T\mid y)$, we first compute the joint distribution of a sequence of observations with the hidden state.

\begin{proposition}    \label{prop:modified_matrix_mult}
Given the initial distribution $\mu_0$,
the joint probability of observing $o$ and arriving at next hidden state $i$ is 
\begin{align*}
 \probs_\theta(O_1=o ,\mathcal{Z}_2 =i)=
\mathbf{1}_i^T A_{o}^\theta\mu_0, 
\end{align*}
where $\mathbf{1}_i$ denotes an $N\times 1$ one-hot vector with the $i$-th entry being assigned value $1$.
\end{proposition}
\begin{proof}
    Consider the right-hand side and from the definition of the observable operators, we have,

\begin{equation}
\begin{aligned}
&\mathbf{1}_i^T A_{o_1}^\theta \mu_0 \\
&= \begin{bmatrix}
   \underbrace{\bm{0}}_{\in R^{1\times (i-1)}}, &  1 ,&  \underbrace{ \bm{0}}_{\in R^{1\times (N-i)}}
\end{bmatrix} \cdot \\
& \begin{bmatrix}
    T_\theta[1,1]\mathbf{E}(o_1|1) & \cdots & T_\theta[1,N]\mathbf{E}(o_1|N) \\
    T_\theta[2,1]\mathbf{E}(o_1|1) & \cdots & T_\theta[2,N]\mathbf{E}(o_1|N) \\
    \vdots & \ddots & \vdots \\
    T_\theta[N,1]\mathbf{E}(o_1|1) & \cdots & T_\theta[N,N]\mathbf{E}(o_1|N) 
\end{bmatrix}\begin{bmatrix}
    \mu_0(1)\\
    \mu_0(2)\\
    \vdots\\
    \mu_0(N)
\end{bmatrix} \\
&= \begin{bmatrix} T_\theta[i,1]\mathbf{E}(o_1|1) & \cdots & T_\theta[i,N]\mathbf{E}(o_1|N)
\end{bmatrix} \begin{bmatrix}
    \mu_0(1)\\
    \mu_0(2)\\
    \vdots\\
    \mu_0(N)
\end{bmatrix} \\
&= T_\theta[i,1]\mathbf{E}(o_1|1)\mu_0(1) + \cdots + T_\theta[i,N]\mathbf{E}(o_1|N)\mu_0(N) \\
&= \sum_{j} \mathbb{P}_\theta(o_1, i|j)\mu_0(j) = \mathbb{P}_\theta(O=o_1, \mathcal{Z}_2=i). 
\end{aligned}
\end{equation}

\end{proof}

\begin{proposition}
    The joint distribution of a sequence of observations and the arrived hidden state is given as
    \[
    \mathbb{P}_\theta (o_{1:t},\mathcal{Z}_{t+1}=i) = \mathbf{1}_i^T A_{o_{t:1}}^\theta \mu_0. 
    \]
\end{proposition}
\label{prop:generalized_joint_dist_matrix_mult}
\begin{proof} By induction.
    From Prop. \ref{prop:modified_matrix_mult}, we have for a single observation case that $\mathbb{P}_\theta(o_1,\mathcal{Z}_{2}=i)=\mathbf{1}_i^T A_{o_1} \mu_0$. Thus, the statement holds for the base case for $t=1$.

    Assume that the statement holds for $t=p$, $p>1$, \ie, 
    \[
    \probs_\theta(o_{1:p}, \mathcal{Z}_{p+1}=i) = \mathbf{1}_i^T A_{o_{p:1}}^\theta \mu_0.
    \]

    Then, in the case of $t=p+1$, by the definition of joint probabilities we have,
    \[
    \mathbb{P}_\theta(o_{1:p+1}, \mathcal{Z}_{p+2}=i)=\sum_{j=1}^N \mathbb{P}_\theta(o_{1:p+1},\mathcal{Z}_{p+1}=j, \mathcal{Z}_{p+2}=i).
    \]
    Applying the chain rule, we have,
    \[
    \mathbb{P}_\theta(o_{1:p+1}, \mathcal{Z}_{p+2}=i)=\sum_{j=1}^N T_\theta[i,j] \cdot \mathbf{E}(o_{p+1}\mid j) \cdot \mathbb{P}_\theta(o_{1:p},j).
    \]
    Therefore,   according to the induction hypothesis, we get
    \begin{align*}
 & \mathbb{P}_\theta(o_{1:p+1}, \mathcal{Z}_{p+2}=i)\\
 = & \sum_{j=1}^N \underbrace{T_\theta[i,j] \cdot \mathbf{E} (o_{p+1}\mid j) }_{A^\theta_{o_{p+1}}[i,j] } \cdot \mathbf{1}_j^T A^\theta_{o_{p:1}} \mu_0 \\
    = &  \sum_{j=1}^N A^\theta_{o_{p+1}}[i,j]\cdot \mathbf{1}_j^T A_{o_{p:1}}^\theta \mu_0\\
    = & \mathbf{1}_{i}^T A_{o_{p+1:1}}^\theta \mu_0.
    \end{align*}
   The last step is   because the step before it represents the summation of the multiplications of the $[i,j]$-th element of the matrix $A^\theta_{p+1}$ with the $j$-th element of vector $A^\theta_{o_{p:1}}\mu_0$. This summation equals to the $i$-th row of the  product of $A^\theta_{o_{p+1}}$ and $A^\theta_{o_{p:1}}\mu_0$.
Therefore, by induction, it holds for all $t$ that $\mathbb{P}_\theta(o_{1:t}, \mathcal{Z}_{t+1}=i)=\mathbf{1}_i^T A_{o_{t:1}}^\theta \mu_0$.

\end{proof}

 Given that $w_T$ takes a binary value, based on the final state of P1 being in one of the goal states, we have the following
\begin{equation}
\label{eq:conditional_wT_y}
\begin{aligned}
    &\probs_\theta(w_T=1\mid o_{1:T}) = \sum_{g\in \mathbf{G}} \frac{\probs_\theta(\mathcal{Z}_T=g, o_{1:T})}{\probs_\theta(o_{1:T})}\\
    &= \sum_{g\in \mathbf{G}} \frac{\mathbf{E}(o_T\mid \mathcal{Z}_T=g)\cdot \probs_\theta(\mathcal{Z}_T=g, o_{1:T-1})}{\probs_\theta(o_{1:T})}\\
\end{aligned}
\end{equation}

To obtain $\probs_\theta(w_T=0\mid o_{1:T})$, we simply compute $\probs_\theta(w_T=0\mid o_{1:T})= 1-\probs_\theta(w_T=1\mid o_{1:T})$. Thus, for the gradient computation, we have, $\nabla_\theta \probs_\theta(w_T=0\mid o_{1:T}) = -\nabla_\theta \probs_\theta(w_T=1\mid o_{1:T})$. Let $y=o_{1:T}$,
\begin{equation}
\label{eq:gradient_conditional_wt_y}
\begin{aligned}
    & \nabla_\theta \probs_\theta(w_T=1\mid y) \\
    & = \nabla_\theta \big[ \sum_{g\in \mathbf{G}} \frac{\mathbf{E}(o_T\mid \mathcal{Z}_T=g)\cdot \probs_\theta(\mathcal{Z}_T=g, o_{1:T-1})}{ \probs_\theta(y)} \big] \\
    & = \big[ \sum_{g\in \mathbf{G}} \mathbf{E}(o_T\mid \mathcal{Z}_T=g)\cdot \nabla_\theta\frac{ \probs_\theta(\mathcal{Z}_T=g, o_{1:T-1})}{ \probs_\theta(y)} \big].
\end{aligned}
\end{equation}

Let $N(\theta,g)= \probs_\theta(\mathcal{Z}_T = g, o_{1:T-1}) = \mathbf{1}_g^T A_{o_{T-1:1}}^\theta \mu_0$. 
Then, using the quotient rule we have,
\begin{equation}
\begin{aligned}
    \label{eq:grad_conditional_wT_y_intermediate}
     & \nabla_\theta \probs_\theta(w_T=1\mid y) \\
     & = \sum_{g\in \mathbf{G}} \mathbf{E}(o_T\mid g) \cdot \frac{\nabla_\theta N(\theta,g)\cdot \probs_\theta(y) - N(\theta,g)\cdot \nabla_\theta \probs_\theta(y)}{(\probs_\theta(y))^2}.
\end{aligned}
\end{equation}

We have $\nabla_\theta \probs_\theta(y)$ from Eq.\ref{eq:gradient_of_P_y} and to obtain the above gradient, we compute $\nabla_\theta N(\theta,g)$ similar to Eq. \ref{eq:gradient_of_P_y}.
Then, with Eq. \ref{eq:grad_conditional_wT_y_intermediate}, we obtain the gradient $\nabla_\theta \probs_\theta(w_T=1\mid y)$.

Finally, by computing $\probs_\theta(y)$, $\nabla_\theta \probs_\theta(y)$, $\probs_\theta (w_T\mid y)$ and $\nabla_\theta \probs_\theta(w_T\mid y)$, we obtain     the gradient of the conditional entropy, $\nabla_\theta H(W_T\mid Y; \theta)$ using Eq. \ref{eq:conditional_entropy_gradient_calculation}. Further, since  the set $\mathcal{O}^T$  of observations  is combinatorial and may be too large to enumerate. We can thus employ sample approximation to estimate the gradient of the conditional entropy. That is, given a set of sequence of observations $\mathcal{Y}=\{y_1,y_2, \cdots, y_V\}$ with $V$ sequences, the approximate conditional entropy is given by,
\begin{equation}    
\label{eq:approximate_entropy}
H(W_T\mid Y;\theta) \approx -\frac{1}{V} \sum_{v=1}^V \sum_{w_T\in \{0,1\}} \probs_\theta (w_T\mid y_v) \log_2 \probs_\theta(w_T\mid y_v).
\end{equation}


\begin{equation}
\label{eq:approximate_gradient_of_entropy}
\begin{aligned}[t]
& \nabla_\theta H(W_T\mid Y;\theta) 
  \approx -\frac{1}{V}\sum_{v=1}^V \sum_{w_T\in \{0,1\}} [ \log_2 \probs_\theta(w_T \mid y) \nabla_\theta \probs_\theta(w_T\mid y) \\
& \quad + \probs_\theta(w_T\mid y) \log_2 \probs_\theta(w_T \mid y) \frac{\nabla_\theta \probs_\theta(y)}{\probs_\theta(y)}   + \frac{\nabla_\theta \probs_\theta(w_T \mid y)}{\log 2} ].
\end{aligned}
\end{equation}

\section{Experimental Validation}
This section shows the effectiveness of the algorithm with two sets of experiments.
The  computations are conducted using Python and PyTorch on a Windows 10 machine with an Intel Core i7 CPU @ 3.2 GHz, 32 GB RAM, and 8 GB RTX 3060 GPU. 
\subsection{The illustrative example}
\begin{example}[Part III]
\label{ex:3}
    We now implement the primal-dual policy gradient algorithm on the illustrative example, Example \ref{ex: 1}. It is recalled that the secret states set is $\{s_4, s_6\}$.
    We employ soft-max policy parametrization, \ie,
    \begin{equation}
    \label{eq:softmax}
        \pi_\theta(\sigma\mid s) = \frac{\exp(\theta_{s, \sigma})}{\sum_{\sigma'\in \Sigma} \exp(\theta_{s, \sigma'})},
    \end{equation}
    where $\theta \in \reals^{|S \times \Sigma|}$ is the policy parameter vector. The soft-max policy has essential analytical properties including completeness and differentiability.
    
    We implemented the primal-dual policy gradient algorithm described above for $1000$ iterations and used $1500$ sample trajectories of length $T=2$ for each iteration. First, we computed the approximate conditional entropy with no masking and then with masking for two different cost budgets $\epsilon=60$ and $\epsilon=20$. The approximate conditional entropy and the approximate gradient for the conditional entropy were computed based on \eqref{eq:approximate_entropy} and \eqref{eq:approximate_gradient_of_entropy}.



   Figure~\ref{fig:running_example_conditional_entropy} shows the conditional entropy for varying budgets of masking costs ($\epsilon$). Recall that a higher entropy implies a higher level of uncertainty.  
   For comparison, we first computed the entropy of the  prior  distribution about the secret $W_T$, which is $0.9172$. This value is the upper bound on the  conditional entropy given observations. Figure~\ref{fig:running_example_conditional_entropy} shows the approximate conditional entropy for no masking drawn in black. It is observed that the conditional entropy is approximately $0.0895$. Thus, without dynamic masking, the observer on average is certain whether or not a secret state is reached.  In Figure~\ref{fig:running_example_conditional_entropy} for $\epsilon = 60$,   the algorithm converges after about $65$ iterations to an approximate conditional entropy of $0.7132$. Likewise, for $\epsilon=20$, the algorithm converges at about $75$ iterations to an approximate conditional entropy of $0.658$.  

    \begin{figure}
    \centering
    \begin{subfigure}{\linewidth}
        \centering
        \includegraphics[width=0.65\linewidth]{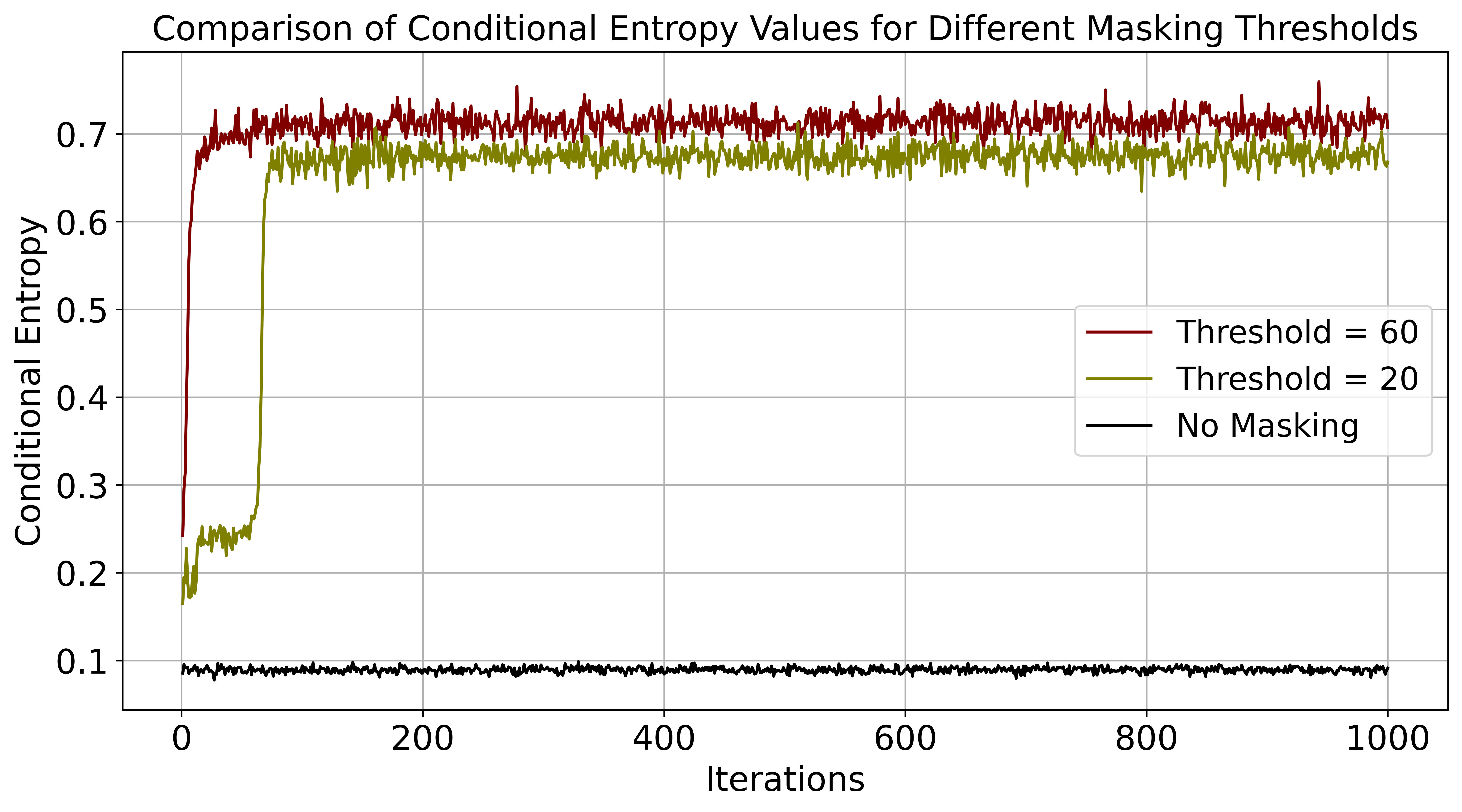}
        \caption{Comparison of conditional entropy.}
        \label{fig:running_example_conditional_entropy}
    \end{subfigure}
    \begin{subfigure}{\linewidth}
        \centering
        \includegraphics[width=0.65\linewidth]{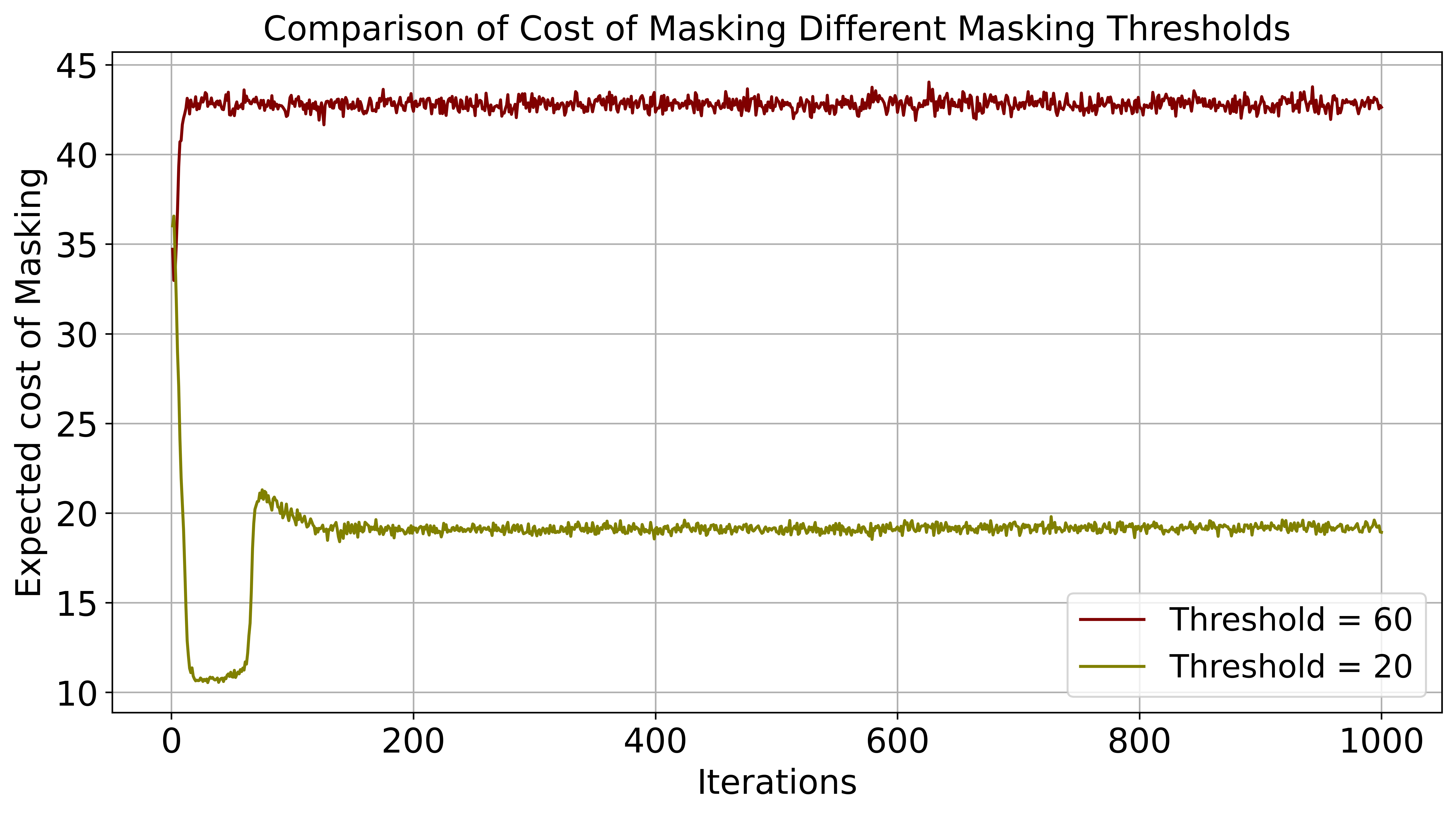}
        \caption{Comparison of the expected cost of dynamic masking.}
        \label{fig:running_example_cost_of_masking}
    \end{subfigure}
    
    \caption{Experimental validation of the illustrative example.}
    \label{fig:mainfig}
\end{figure}
    
    Figure~\ref{fig:running_example_cost_of_masking} shows the expected cost of the dynamic masking for thresholds $60$ and $20$. It can be seen that as the algorithm converges, the cost of masking converges to a value within the given threshold budget, \ie,  $42.63$ and $18.93$ for $\epsilon = 60$, and $\epsilon = 20$, respectively. The average time consumed for each iteration in both cases is $5.7$s.

    From the two comparisons, the conditional entropy improves by using dynamic masking. It is also observed that the conditional entropy decreases with a decrease in the budget, as expected.
\end{example}

\subsection{Case Study: Stochastic gridworld}

Consider a secure pharmaceutical research facility responsible for testing and ensuring the safety of various drug batches before they are released to the market represented by the $6 \times 6$ grid world as shown in Figure~\ref{fig:pharma_environment}.

\begin{figure}[h]
  \centering
  \includegraphics[width=0.4\linewidth]{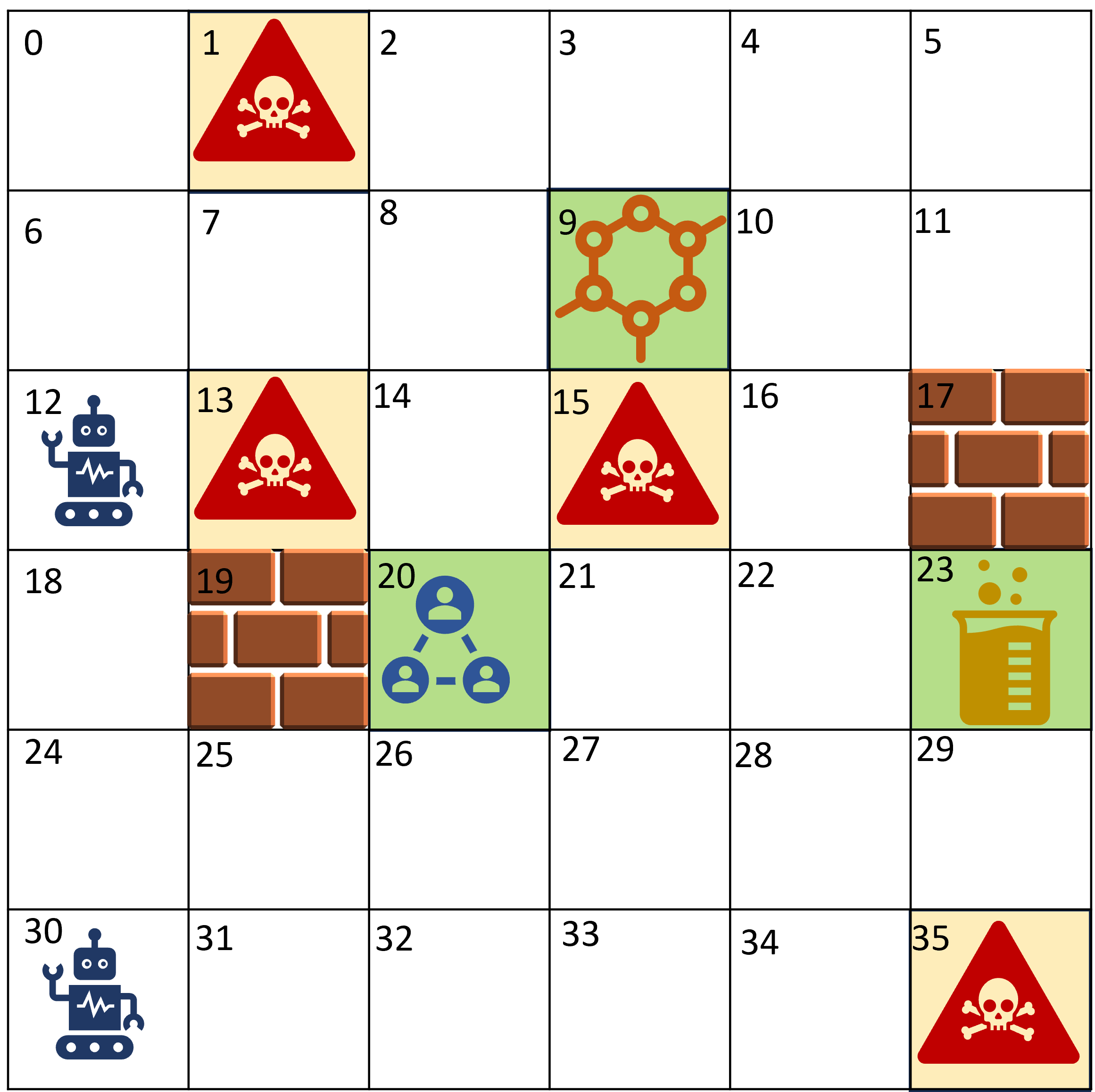}
  \caption{Gridworld depiction of a pharmaceutical research facility with cells $\{12, 30\}$ as the initial states for the robot.}
  \Description{Enter a description.}
  \label{fig:pharma_environment}
\end{figure}

In this scenario, the robot is tasked with supplying essential equipment to one of the following locations: the chemical analysis lab in cell $23$, the biological activity testing chamber in cell $9$, or the human testing and analysis lab in cell $20$. The robot also must steer clear of the bio-hazard disposal units in cells $1, 13, 15$, and $35$. The goal states of the robot, cells $9, 20, 23$, are also the secret states. Cells $17$ and $19$ are walls.
A competitor company intends to know the progress on the breakthrough drugs being designed by the research facility. Thus, it is a passive observer (P2) that observes the robot's behavior using a deployed sensor network, specifically, observes the final state of the robot. We have a masking agent (P1) in the environment that aims to maximize the final state opacity to P2. 

\paragraph*{Environment dynamics} The robot can move in one of four compass directions (North, South, East, West) within the grid world. Every action of the robot carries a degree of uncertainty. That is, when the robot moves in a specific direction, it can reach the intended cell with a probability $p$ and reach the unintended cells (\ie, if north is the intended cell, the cells in east and west are the unintended cells) with a probability $(1-p)/2$.
If the robot moved into a cell with a wall or out of the boundaries, it would remain in the current cell. For instance, if the robot moves East from cell $30$, it reaches cell $31$ with probability $p=0.8$, cell $24$ with probability $0.1$, and remains in cell $30$ with probability $0.1$.

In this example, the robot starts uniformly randomly from $12$ or $30$. The robot is aware of its state in the environment and follows a pre-planned goal policy, shown in Figure~\ref{fig:goal_policy_and_sensor_setup} using the blue arrows. Any cell with multiple arrows implies that the robot uniformly randomly chooses from one of those control actions. P2 observes partially the robot's trajectories through the sensors $A, B, C$, and $D$ deployed in the environment. The sensors cover cells $\{3, 4, 9, 10\}$, $\{21, 22, 28\}$, $\{23, 29, 35\}$, and $\{6, 7, 8, 12, 13, 14\}$ respectively as shown in Fig. \ref{fig:goal_policy_and_sensor_setup}. The sensors are considered binary sensors that return the sensor labels $A$, $B$, $C$, or $D$ if the robot is in their coverage and a null observation `0' otherwise. We consider the sensors to have a default detection probability $\beta$. Thus, sensors detect the robot in their coverage with a probability $\beta$ and not detect with a probability $1-\beta$ (\ie, the false negative probability). The false positive probability is zero for each sensor. We vary the value of $\beta$ and analyze the impact of the sensor's noise on optimizing opacity.

\begin{figure}[h]
  \centering
  \includegraphics[width=0.5\linewidth]{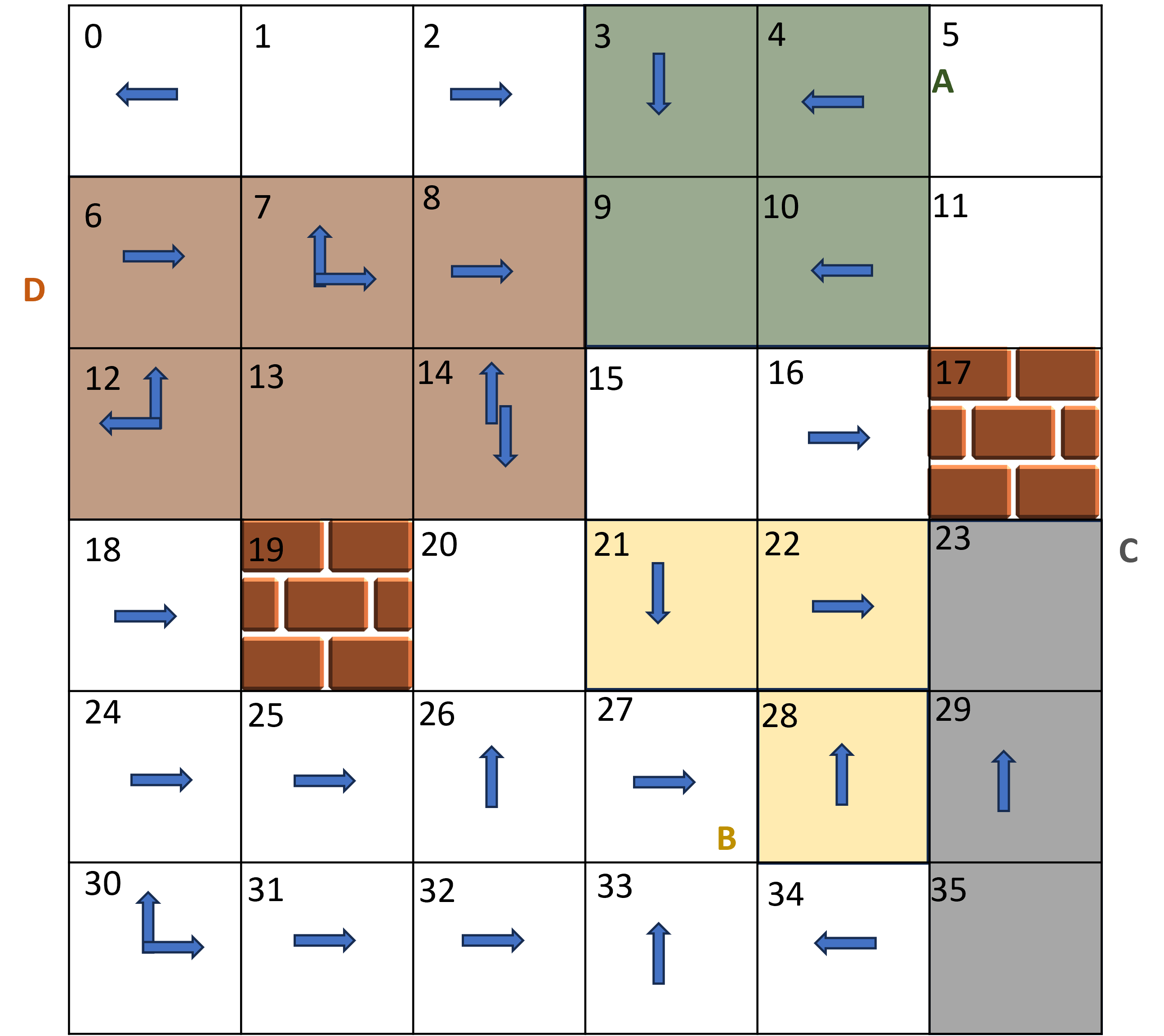}
  \caption{The sensor setup with sensors $A$ (Green), $B$ (Yellow), $C$ (Grey), and $D$ (Brown), and the robot's goal policy (Blue arrows) in the grid world environment.}
  \Description{Enter a description.}
  \label{fig:goal_policy_and_sensor_setup}
\end{figure}

\paragraph*{Experiment setup}
In our experiments, we consider the dynamic masking agent can mask any one sensor or none at any given time, \ie, $\Sigma = \{A, B, C, D, N\}$, where $N$ denotes no sensor being masked. We also consider a strong P2 to whom the dynamic masking actions are visible. Thus, the product of sensor readings and the masking actions form the set of observations. 
With the robot's goal policy, the system under partial observation is modeled as the \ac{hmm} $M$. Finally, we have the dynamic masking agent pays the cost of masking each sensor, $\phi(A) = 20$, $\phi(B) = 25$, $\phi(C) = 15$ and $\phi(D) = 10$, as follows.
\[
C(s,\sigma, \sigma') = \begin{cases}
    \phi(\sigma') &\text{if } \quad \sigma \neq \sigma', \sigma'\in \Sigma\setminus\{N\},\\
       \frac{\phi(\sigma')}{2} & \text{if } \quad \sigma = \sigma', \sigma' \in \Sigma\setminus\{N\},\\
       0& \text{ otherwise.}
\end{cases}
\]
That is, the cost of masking a previously masked sensor is considered to be half of its original cost. 
With this setup, we conduct the experiments for the threshold values $\epsilon = 70$ and $\epsilon = 35$ for varying sensor stochasticity parameters $\beta = 0.85$ and $\beta = 0.75$. We consider a finite horizon   $T=10$. Similar to Example \ref{ex:3}, we again employ the soft-max parametrization as in Eq. \ref{eq:softmax}, with $\theta$ being the policy parameter vector.
In each iteration, $1500$ trajectories of length $10$ are sampled. These trajectories are divided into $15$ batches to update $\theta$ in the primal-dual policy gradient computation.  The algorithm took about 2000 iterations to converge and each iteration of the policy gradient algorithm takes an average of $14.026$s. 

 The entropy of the prior about the final-state secret $W_T$ before receiving any observation in this environment is $0.9042$. To illustrate the effectiveness of masking, 
we compare the synthesized mask with (1) the no-masking policy, reflecting P2's uncertainty due to the environment and sensor setup, and (2) the final state masking policy, which masks the sensor covering the secret state when the probability of the robot's next state being secret is non-zero.
\paragraph*{Discussion}
For each policy, we use trajectories to compute a sample approximation of the conditional entropy under that policy as per Eq.~\ref{eq:approximate_entropy}. Table~\ref{tab:entropy_and_cost_comparison} summarizes the approximate conditional entropy values for the baseline policies and our constrained optimal masking policy for varying the values of $\beta=0.85, 0.75$ and for varying thresholds $\epsilon=70, 35$.  

\begin{figure}
    \centering
    \begin{subfigure}{\linewidth}
        \centering
        \includegraphics[width=0.65\linewidth]{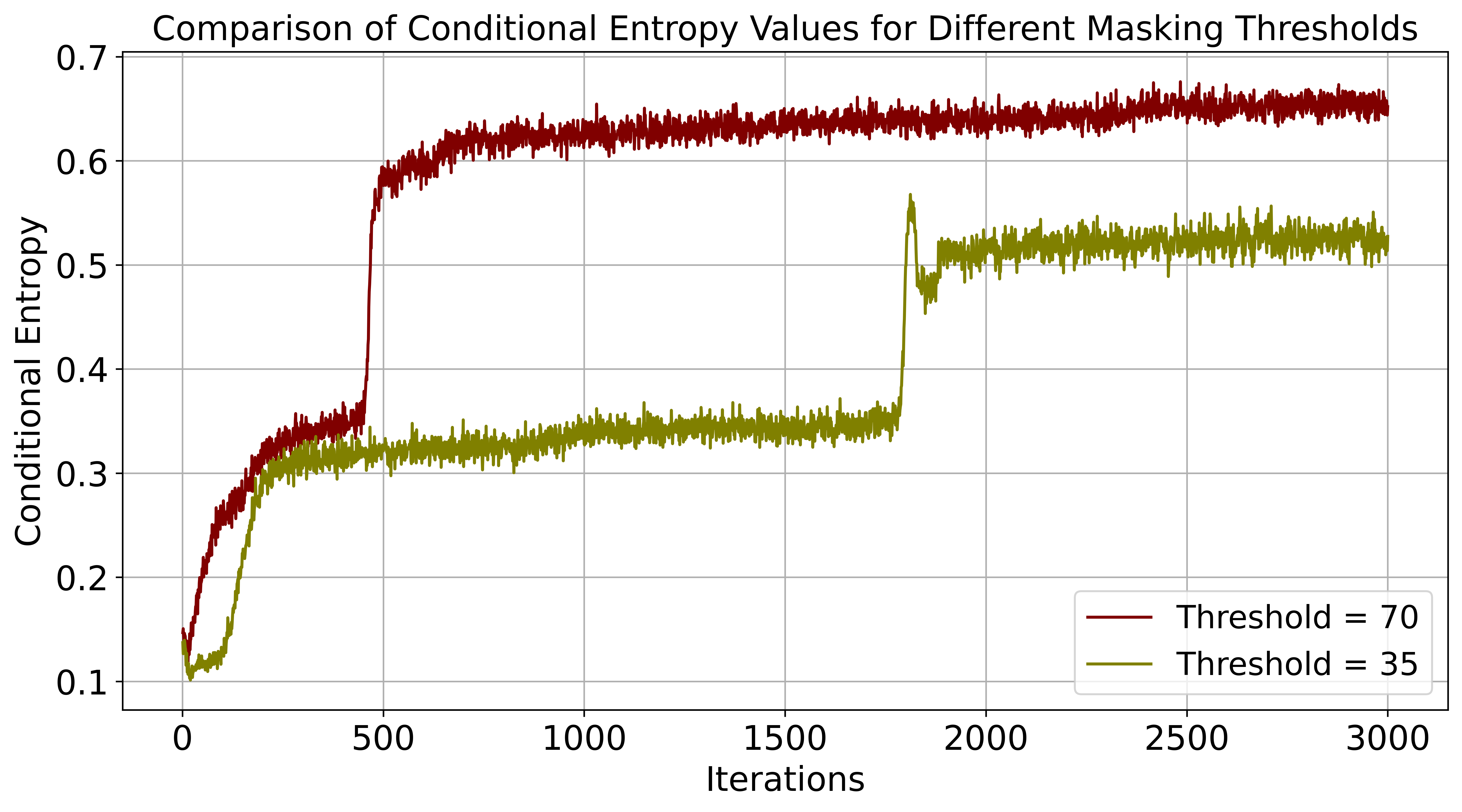}
        \caption{Comparison of conditional entropy for $\beta = 0.85$.}
        \label{fig:gridworld_example_conditional_entropy_0.85}
    \end{subfigure}
    \begin{subfigure}{\linewidth}
        \centering
        \includegraphics[width=0.65\linewidth]{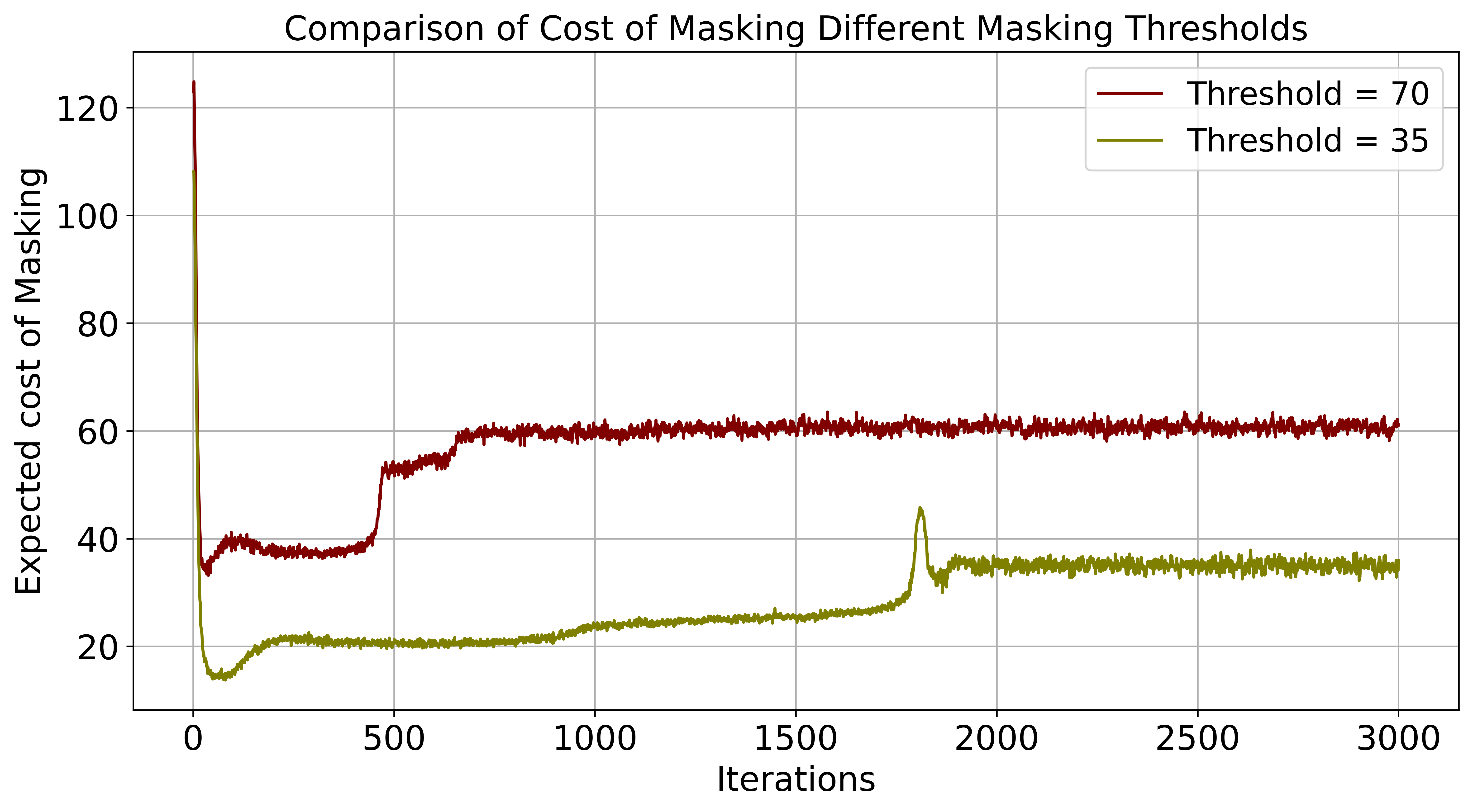}
        \caption{Comparison of the expected cost of dynamic masking ($\beta = 0.85$).}
        \label{fig:gridworld_example_cost_of_masking_0.85}
    \end{subfigure}

    \begin{subfigure}{\linewidth}
        \centering
        \includegraphics[width=0.65\linewidth]{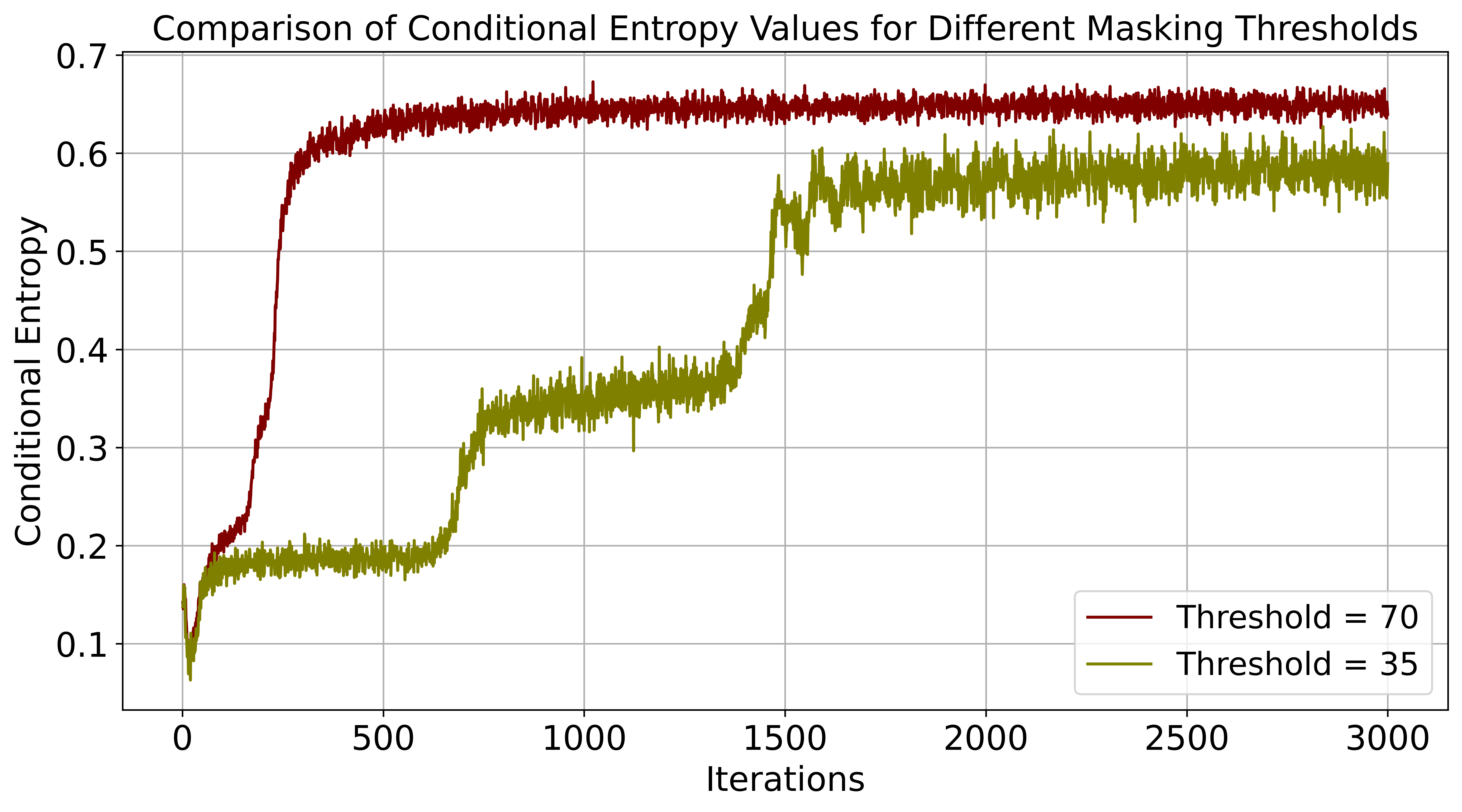}
        \caption{Comparison of conditional entropy for $\beta = 0.75$.}
        \label{fig:gridworld_example_conditional_entropy_0.75}
    \end{subfigure}

    \begin{subfigure}{\linewidth}
        \centering
        \includegraphics[width=0.65\linewidth]{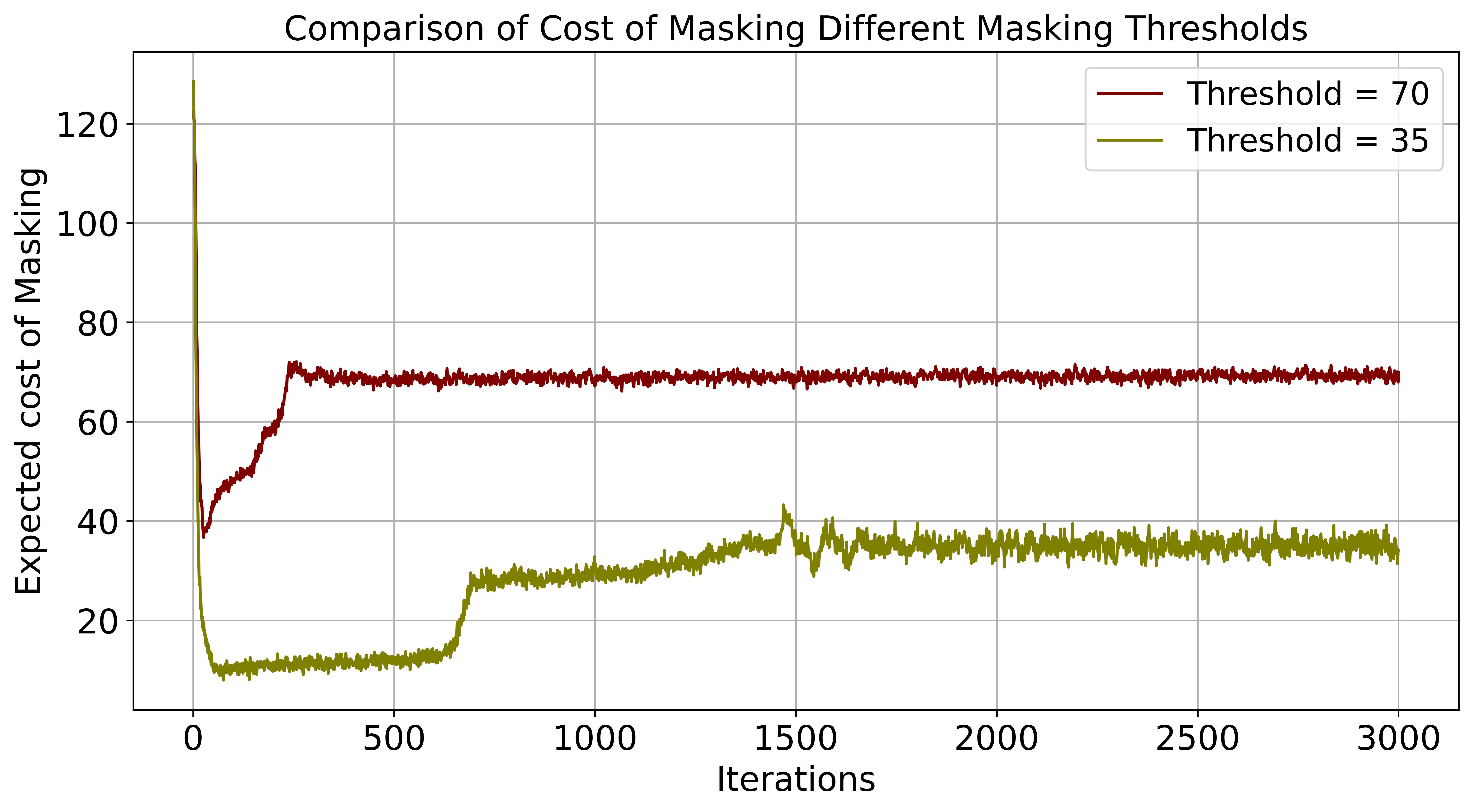}
        \caption{Comparison of the expected cost of dynamic masking ($\beta = 0.75$).}
        \label{fig:gridworld_example_cost_of_masking_0.75}
    \end{subfigure}
    
    \caption{Convergence results for the grid world example.}
    \label{fig:mainfig-2}
\end{figure}


\begin{table}[h]
    \centering
    \begin{tabular}{|c|c|c|c|}
        \hline
        Policy & $\epsilon$ & Approx. $H(W_T\mid O_{0:T})$ & Exp. cost \\
        \hline
        No-masking & - & $0.168$ ($\beta=0.85$)& $0$ \\
        No-masking & - & $0.1824$ ($\beta=0.75$) & $0$ \\
        Final state masking & - & $0.1763$ ($\beta=0.85$) & $14.37$ \\
        Final state masking & - & $0.2129$ ($\beta=0.75$) & $14.8$ \\
        Our policy  & $70$ & $0.6539$ ($\beta=0.85$) & $61.37$ \\
        Our policy  & $35$ & $0.5274$ ($\beta=0.85$) & $34.088$ \\
        Our policy  & $70$ & $0.6543$ ($\beta=0.75$) & $69.86$ \\
        Our policy  & $35$ & $0.5893$ ($\beta=0.75$) & $34.289$ \\
        \hline
    \end{tabular}
    \caption{Approximate conditional entropy and expected cost for different policies under varying $\epsilon$ and $\beta$.}
    \label{tab:entropy_and_cost_comparison}
\end{table}

Given the sensor stochasticity parameter $\beta=0.85$, Figures~\ref{fig:gridworld_example_conditional_entropy_0.85} and \ref{fig:gridworld_example_conditional_entropy_0.75} show the convergence results of primal-dual policy gradient algorithms for solving constrained optimal masking policies with thresholds $\epsilon = 70$ and $\epsilon = 35$. The conditional entropy of the policy converges to $0.6539$ and $0.5274$, respectively. Similarly, for $\beta=0.75$, it converges to $0.6543$ and $0.5893$ for $\epsilon = 70$ and $\epsilon = 35$, respectively, as shown in Figures~\ref{fig:gridworld_example_conditional_entropy_0.75} and \ref{fig:gridworld_example_cost_of_masking_0.75}. Experiments show that larger $\epsilon$ increases conditional entropy, and for a given $\epsilon$, a higher false negative rate $1-\beta$ slightly raises entropy. Table~\ref{tab:entropy_and_cost_comparison} also shows the synthesized dynamic mask outperforms both baseline policies.



To gain some insight into the synthesized dynamic mask, consider the following trajectory $30\to24\to25\to26\to20$. 


Our policies (for both $\epsilon=70$ and $\epsilon=35$) choose to mask no sensor with a high probability ($0.9915$) at the state $30$. 
At state $24$, the masking action $D$ is chosen with a probability $0.6782$, masking action $N$ with a probability $0.303$, and actions $B$, and $C$ with remaining very low probabilities. It continues to mask the sensor $D$ with probabilities $0.9851$, $0.9847$, and $0.9711$ in the states $25$, $26$, and $20$ (the other masking actions are chosen with negligible probabilities). 

These masking decisions are interesting as they show that the agent strategically exploits the fact that its masking actions are visible to the observer. Consider the state $25$. From the observer's perspective, given his null observations so far, the robot could be in any of the following states due to the environment's stochasticity: $\{30, 24, 18, 25, 31, 32\}$.
At state $25$, the robot took the control action `East' (as shown in Figure~\ref{fig:goal_policy_and_sensor_setup}), and reached state $26$ in this trajectory. After this one step, the observer believes the robot could end up in any of the states, \ie, $\{30, 24, 18, 25, 31, 32, 26, 33, 12\}$. Before receiving an observation, the observer is uncertain whether the robot is in a region not covered by any sensors or if it has moved to state $12$, which is under sensor $D$. 

Now, if the masking agent were to mask any sensor other than $D$, say $B$, the observer would receive the observation (`0', $B$). Because $D$ is not masked, it emits a ``null'' observation. This observation would immediately eliminate the possibility that the robot is in state $12$. However, as $D$ is masked,  the observer remains uncertain about whether the robot is in a state covered by sensor $D$ or in the uncovered region because the observer cannot distinguish whether the null observation is due to the masking of sensor $D$ or the robot's actual position being outside the sensor's coverage. Hence, the masking action $D$ incurs a higher level of uncertainty in the observer's belief compared to other masking actions.

Alternatively, consider a trajectory where the robot did not reach $20$ and followed the sequence $26\to 27 \to 28 \to 29$. The policy generated for $\epsilon=35$ masks sensor $D$ with probability $0.9715$ in state $27$. Then, the robot reached state $28$ and the chosen masking policy is to mask sensor $B$ with probability $0.3774$ and sensor $C$ with probability $0.328$. Similarly, when the robot is at the state $29$, it chooses to mask the sensor $C$ with a probability of $0.9984$. The policy for $\epsilon=70$ masks sensor $D$ with probability $0.6943$ in state $27$. Then, in state $28$, the preferred masking action is $B$ with probability $0.6837$ and $C$ with probability $0.2149$. 

Here, we observe how the masking cost influences the selection of masking actions. For example, at state $28$, sensor $B$ — the most expensive sensor — 
is the best choice for masking as the control action taken is `North.' It
is chosen as the masking action with a probability of $0.3774$ when the threshold is set to $35$. However, this probability increases by nearly $0.3$ when the budget is raised, highlighting the impact of the available budget on masking decisions.





\section{Conclusion and Future Work}
In this paper, we introduce conditional entropy as a measure of opacity for a stochastic system. We formulate the problem of synthesizing a budget-constrained dynamic mask policy, to maximize final-state opacity. To enforce the budget costs on masking actions, we develop a primal-dual policy gradient algorithm and derive a procedure to compute the gradient of the conditional entropy with respect to the masking policy parameters using observation operators. Our experimental results demonstrate the effectiveness of the proposed algorithm against two baseline policies, revealing interesting characteristics of the generated policy.

Our work extends information-flow security objectives to dynamical systems against passive observers. The methods can be applied to design dynamic information-releasing policies in a cyber-physical system. In future research, one promising direction is to generalize the proposed method for other classes of information-theoretic opacity properties, including language-based opacity or infinite-step opacity. Another interesting direction is to consider the trade-off between transparency and opacity, for example, the system must release information to make certain properties to be legible to the external observer while ensuring privacy and opacity of other confidential properties. 



\section*{Acknowledgements}

This work was sponsored in part by the Army Research Office and was accomplished under Grant
Number W911NF-22-1-0034 and by the Army Research Laboratory under Cooperative Agreement Number W911NF-22-2-0233, and in part by NSF under grant No. 2144113. The views and conclusions contained in this document are those of the authors and
should not be interpreted as representing the official policies, either expressed or implied, of the Army Research
Office, the Army Research Lab, or the U.S. Government. The U.S. Government is authorized to reproduce and distribute reprints for
Government purposes notwithstanding any copyright notation herein.

\bibliographystyle{ACM-Reference-Format}
\bibliography{refs}

\end{document}